\documentclass[journal]{IEEEtran}

\usepackage{latexsym}
\usepackage{graphicx}
\usepackage{array}
\usepackage{amsmath}
\usepackage{amsfonts}
\usepackage{amssymb}
\usepackage{amsthm}
\usepackage{algpseudocode}
\usepackage{algorithm}

\newtheorem{thm}{Theorem}
\newtheorem{lemma}{Lemma}
\newtheorem{prop}{Proposition}

\newtheorem{cor}{Corollary}

\newcommand{\bx} {\boldsymbol{x}}

\newcommand{\by} {\boldsymbol{y}}

\newcommand{\bH} {\boldsymbol{H}}
\newcommand{\bB} {\boldsymbol{B}}
\newcommand{\bA} {\boldsymbol{A}}

\newcommand{\bI} {\boldsymbol{I}}
\newcommand{\bR} {\boldsymbol{R}}
\newcommand{\bU} {\boldsymbol{U}}
\newcommand{\bu} {\boldsymbol{u}}

\newcommand{\bW} {\boldsymbol{W}}
\newcommand{\bQ} {\boldsymbol{Q}}
\newcommand{\bM} {\boldsymbol{M}}
\newcommand{\bP} {\boldsymbol{P}}

\newcommand{\bLam} {\boldsymbol{\Lambda}}
\newcommand{\gl}{\lambda}

\newcommand{\bxi} {\boldsymbol{\xi}}

\newcommand{\sN} {\mathcal{N}}
\newcommand{\sR} {\mathcal{R}}
\newcommand{{\diag}} {\mathrm{diag}}

\def\bal#1\eal{\begin{align}#1\end{align}}
\newcommand{\bp} {\begin{proof}}
\newcommand{\ep} {\end{proof}}

\newcommand{{\bRF}} {\right\}}

\newcommand{\tr}{\operatorname{tr}}

\begin{document}

\title{On The Capacity of Gaussian MIMO Channels Under Interference Constraints (full version)}

\author{Sergey Loyka

\vspace*{-1\baselineskip}


\thanks{S. Loyka is with the School of Electrical Engineering and Computer Science, University of Ottawa, Ontario, Canada, e-mail: sergey.loyka@uottawa.ca.}

}

\maketitle

\vspace*{-1\baselineskip}
\begin{abstract}
Gaussian MIMO channel under total transmit and multiple interference power constraints (TPC and IPCs) is considered. A closed-form solution for its optimal transmit covariance matrix is obtained in the general case. A number of more explicit closed-form solutions are obtained in some special cases, including full-rank and rank-1 (beamforming) solutions, which differ significantly from the well-known water-filling based solutions (e.g. signaling on the channel eigenmodes is not optimal anymore and the capacity can be zero for non-zero transmit power). Capacity scaling with transmit power is studied: its  qualitative behaviour is determined by a natural linear-algebraic structure of MIMO channels of multiple users. A simple rank condition is given to characterize the cases where spectrum sharing is possible. An interplay between the TPC and IPCs is investigated, including the transition from power-limited to interference-limited regimes. A bound on the rank of optimal covariance is established. A number of unusual properties of optimal covariance matrix are pointed out.
\end{abstract}

\vspace*{-.8\baselineskip}

\section{Introduction}

Aggressive frequency re-use and hybrid (non-orthogonal) access schemes envisioned as key technologies in 5G systems \cite{Shafi-17} can potentially generate significant amount of inter-user interference and hence should be designed and managed carefully. In this respect, multi-antenna (MIMO) systems have significant potential due to their significant signal processing capabilities, including interference cancellation and precoding, which can also be done in an adaptive and distributed manner \cite{Biglieri}\cite{Heath-19}.
The capacity and optimal signalling for the Gaussian MIMO channel under the total power constraints (TPC) is well-known: the optimal (capacity-achieving) signaling is Gaussian and, under the TPC, is on the eigenvectors of the channel with power allocation to the eigenmodes given by the water-filling (WF) \cite{Biglieri}-\cite{Telatar-95}. Under per-antenna power constraints (PAC), in addition or instead of the TPC, Gaussian signalling is still optimal but not on the channel eigenvectors anymore so that the standard water-filling solution over the channel eigenmodes does not apply \cite{Vu-11}\cite{Loyka-17}.

Much less is known under the added interference power constraint (IPC), which limits the power of interference induced by a secondary transmitter to a primary receiver in a spectrum-sharing system. A game-theoretic approach to this problem was proposed in \cite{Scurati-10}, where a fixed-point equation was formulated from which the optimal transmitt covariance matrix can in principle be determined. Unfortunately, no closed-form solution is known for this equation and the considered settings require the channel to the primary receiver to be full-rank hence excluding the important cases where the number of Rx antennas is less than the number of Tx antennas (typical for massive MIMO downlink); the TPC is not included explicitly (rather, being "absorbed" into the IPC), hence eliminating the important case of inactive IPC and, consequently, no interplay between the TPC and the IPC can be studied.

Cognitive radio MIMO systems under interference constraints have been also studied in \cite{Yang-13}\cite{Huh-10}\cite{Zhang-12}, where a number of numerical optimization algorithms were developed but no closed-form solutions are known to the underlying optimization problems. Optimal signaling for the Gaussian MIMO channel under the TPC and the IPC has been also studied in \cite{Zhang-10}-\cite{Loyka-17-2} using the dual problem approach, and was later extended to multi-user settings in \cite{Liu-12}. However, constraint matrices are required to be full-rank and no closed-form solution was obtained for optimal dual variables. Hence, various numerical algorithms or sub-optimal solutions were proposed. This limits insights significantly.

In this paper, we study the spectrum-sharing potential of Gaussian MIMO channels and concentrate on analysis rather than numerical algorithms. This provides deeper understanding of the problem and a number of insights unavailable from numerical algorithms alone. Specifically, we obtain novel closed-form solutions for an optimal transmit covariance matrix for the Gaussian MIMO channel under the TPC and multiple IPCs. All constraints are included explicitly and hence anyone is allowed to be inactive. This allows one to study the interplay between the power and interference constraints and, in particular, the transition from power-limited to interference limited regimes as the Tx power increases. As an added benefit, no limitations is placed on the rank of the channel to the PR, so that the number of antennas of the PR can be any (including massive MIMO settings). In some cases, our KKT-based approach leads to closed-form solutions for the optimal dual variables as well, including full-rank and rank-1 (beamforming) solutions and the conditions for their optimality. A simple rank condition is given to characterize the cases where spectrum sharing is possible for any interference power constraint. In general, the primary user has a major impact on the capacity at high SNR while being negligible at low SNR. The high-SNR behaviour of the capacity is qualitatively determined by the null space of the PR's channel matrix. The presented closed-form solutions of optimal signaling can be used directly in massive MIMO settings. Since numerical complexity of generic convex solvers can be prohibitively large for massive MIMO (in general, it scales as $m^6$ with the number $m$ of antennas), the above analytical solutions are a valuable low-complexity alternative.

Under the added IPC(s), the unitary-invariance of the feasible set is lost and hence many known solutions and standard "tricks" (e.g. Hadamard inequality) of the analysis under the TPC alone cannot be used. This has profound impact on optimal signaling strategies as well as on analytical techniques to solve the underlying optimization problem. In particular, unlike the standard water-filling solution, (i) signaling on the channel eigenmodes is not optimal anymore (unless all IPCs are inactive or if their channel eigenmodes are the same as those of the main MIMO channel); (ii) the rank of an optimal covariance matrix can exceed that of the channel; (iii) an optimal covariance matrix is not necessarily unique; (iv) the channel capacity can be zero for a non-zero Tx power and channel; (v) the channel capacity may stay bounded under unbounded growth of the Tx power. All these phenomena have major impact on the spectrum-sharing capabilities of MIMO channels. We demonstrate that capacity scaling with Tx power under multiple IPCs can be understood in terms of a natural linear-algebraic structure of MIMO channels of different users.

\textit{Notations}: bold capitals ($\bR$) denote matrices while bold lower-case letters ($\bx$) denote column vectors; $\bR^+$ is the Hermitian conjugation of $\bR$; $\bR \ge 0$ means that $\bR$ is positive semi-definite; $|\bR|,\ tr(\bR),\ r(\bR)$ denote determinant, trace and rank of $\bR$, respectively; $\gl_i(\bR)$ is $i$-th eigenvalue of $\bR$; unless indicated otherwise, eigenvalues are in decreasing order, $\gl_1\ge \gl_2\ge ..$; $\lceil\cdot\rceil$ denotes ceiling, while $(x)_+=\max[0,x]$ is the positive part of $x$; $\mathcal{R}(\bR)$ and $\mathcal{N}(\bR)$ denote the range and null space of $\bR$ while $\bR^{\dag}$ is its Moore-Penrose pseudo-inverse; $\mathbb{E}\{\cdot\}$ is statistical expectation.


\vspace*{-.91\baselineskip}
\section{Channel Model}
\label{sec.Channel Model}

Let us consider the standard discrete-time model of the Gaussian MIMO channel:
\bal
\label{eq.ch.mod}
\by_1 = \bH_1\bx +\bxi_1
\eal
where $\by_1, \bx, \bxi_1$ and $\bH_1$ are the received and transmitted signals, noise and channel matrix. This is illustrated in Fig. 1. The noise is assumed to be complex Gaussian with zero mean and unit variance, so that the SNR equals to the signal power. A complex-valued channel model is assumed throughout the paper, with full channel state information available both at the transmitter and the receiver. Gaussian signaling is known to be optimal in this setting \cite{Biglieri}-\cite{Telatar-95} so that finding the channel capacity $C$ amounts to finding an optimal transmit covariance matrix $\bR$, which can be expressed as the following optimization problem (P1):
\bal
\label{eq.C.def}
(P1):\ C = \max_{\bR \in S_R} C(\bR)
\eal
where $C(\bR) = \log|\bI +\bW_1\bR|$, $\bW_1=\bH_1^+\bH_1$, $\bR$ is the Tx covariance and $S_R$ is the constraint set. In the case of the total power constraint (TPC) only, it takes the form
\bal
S_R= S_{TPC} \triangleq \{\bR: \bR\ge 0, \tr(\bR) \le P_T\},
\eal
where $P_T$ is the maximum total Tx power. The solution to this problem is well-known: optimal signaling is on the eigenmodes of $\bW_1$, so that they are also the eigenmodes of optimal covariance $\bR^*=\bR_{WF}$, and the optimal power allocation is via the water-filling (WF). This solution can be compactly expressed as follows:
\bal\notag
\bR_{WF} \triangleq(\mu^{-1}\bI-\bW_1^{-1})_+ = \sum_{i: \gl_{1i} >\mu} (\mu^{-1} -\gl_{1i}^{-1})\bu_{1i}\bu_{1i}^+
\eal
where $\mu\ge 0$ is the "water" level found from the total power constraint $tr(\bR^*)=P_T$, $\gl_{1i}, \ \bu_{1i}$ are $i$-th eigenvalue and eigenvector of $\bW_1$, so that $\gl_{1i} =0$ are excluded from the summation due to $\gl_{1i} >\mu \ge 0$; $(\bR)_+$ denotes positive eigenmodes of Hermitian matrix $\bR$: $(\bR)_+=\sum_{i: \gl_i>0} \gl_i\bu_i\bu_i^+$, where $\gl_i,\ \bu_i$ are $i$-th eigenvalue and eigenvector of $\bR$.

In a spectrum-sharing multi-user system, there is a limit on how much interference the Tx can induce (via $\bx$) to primary user $U_k$, see Fig. \ref{fig.muli-user},
\bal
\mathbb{E}\{\bx^+\bH_{2k}^+\bH_{2k}\bx\} = tr(\bH_{2k}\bR\bH_{2k}^+) \le P_{Ik}
\eal
where $P_{Ik}$ is the maximum acceptable interference power and the left-hand side is the actual interference power at user $U_k$. In this setting, the constraint set becomes
\bal
\label{eq.SR.2}
S_R=\{\bR\ge 0: \ tr(\bR) \le P_T,\ tr(\bW_{2k}\bR) \le P_{Ik}\ \forall k\},
\eal
where $\bW_{2k}=\bH_{2k}^+\bH_{2k}$ and $P_{Ik}$ represent channel to user $k$ and respective interference constraint power, $k=1..K$, $K$ is the number of users, see Fig. \ref{fig.muli-user}.

The Gaussian signalling is still optimal in this setting and the capacity subject to the TPC and IPCs can still be expressed as in \eqref{eq.C.def} but the optimal covariance is not $\bR_{WF}$ anymore. In particular, the unitary-invariance of the feasible set $S_{TPC}$ under the TPC alone is lost due to the presence of the IPCs $tr(\bW_{2k}\bR) \le P_{Ik}$ in $S_R$ so that well-known results and "tricks" (based on unitary invariance of the feasible set) cannot be used anymore. Since the "shape" of the feasible set $S_R$ affects significantly optimal $\bR$, this results in a number of new properties of optimal signaling and of the capacity, as we show below.

\begin{figure}[t]
	\centerline{\includegraphics[width=3.5in]{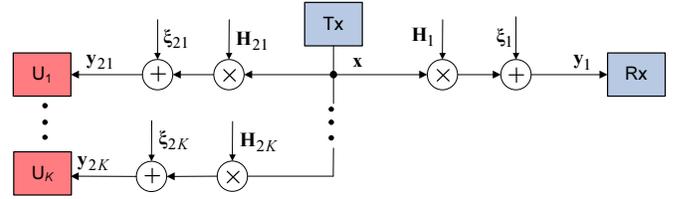}}
	\caption{A block diagram of multi-user Gaussian MIMO channel under interference constraints. $\bH_1$ and $\bH_{2k}$ are the channel matrices to the Rx and $k$-th user respectively. Interference constraints are to be satisfied for each user.}
		\label{fig.muli-user}
\end{figure}


One may also consider the total (rather than individual) interference power constraint so that
\bal\notag
S_{RT}=\{\bR: \bR\ge 0,\ tr(\bR) \le P_T,\ \sum_k tr(\bW_{2k}\bR) \le P_{I} \}
\eal
In this case, all the results of this paper will apply with $K=1$, $P_{I1}=P_I$, and  $\bW_{21} \to \sum_k \bW_{2k}$.

\section{Optimal Signalling Under the TPC and IPCs}
\label{sec.C.Int}

To characterize fully the capacity, a closed-form solution for the optimal signaling problem (P1) in \eqref{eq.C.def} under the joint constraints in \eqref{eq.SR.2} is given below in the general case, i.e. $\bW_1, \bW_{2k}$ are allowed to be singular and any of the constraints are allowed to be inactive. This extends the known results in \cite{Zhang-10}-\cite{Loyka-17-2} to the general case.

\begin{thm}
\label{thm.R*}
Consider the capacity of the Gaussian MIMO channel in \eqref{eq.C.def} under the joint TPC and IPC in \eqref{eq.SR.2}. The optimal Tx covariance matrix to achieve the capacity can be expressed as follows:
\bal
\label{eq.thm.R*.1}
\bR^* = \bW_{\mu}^{\dag} (\bI-\bW_{\mu}\bW_1^{-1}\bW_{\mu})_+ \bW_{\mu}^{\dag}
\eal
where $\bW_{\mu}=(\mu_1\bI+ \sum_k \mu_{2k}\bW_{2k})^{\frac{1}{2}}$; $\bW_{\mu}^{\dag}$ is the Moore-Penrose pseudo-inverse of $\bW_{\mu}$; $\mu_1, \mu_{2k}\ge 0$ are Lagrange multipliers (dual variables) responsible for the TPC and IPCs, found from
\bal
\label{eq.thm.R*.2}
\mu_1(tr(\bR^*) - P_T)=0,\ \mu_{2k}(tr(\bW_{2k}\bR^*)-P_{Ik})=0
\eal
subject to $tr(\bR^*)\le P_T$, $tr(\bW_{2k}\bR^*) \le P_{Ik}\ \forall k$. The respective capacity is
\bal
\label{eq.C}
C = \sum_{i: \gl_{\mu i}>1} \log \gl_{\mu i}
\eal
where $\gl_{\mu i} = \gl_i(\bW_{\mu}^{\dag}\bW_1\bW_{\mu}^{\dag})$.
\end{thm}
\vspace*{-.5\baselineskip}
\begin{proof}
See Appendix.
\end{proof}

Based on \eqref{eq.thm.R*.1}, one observes that $\bW_{\mu}$ plays a role of a "whitening" filter, which disappears when all IPCs are inactive.
When $\bW_1$ is full-rank, i.e. $\bW_1 >0$, then $\bR^*$ is unique, which is not necessarily the case in general - a remarkable difference to the TPC-only case, where $\bR^* =\bR_{WF}$ is always unique.

A number of known special cases follow from \eqref{eq.thm.R*.1}: If $K=1$ and $\bW_{\mu}$ is full-rank, then $\bW_{\mu}^{\dag}=\bW_{\mu}^{-1}$ (see e.g. \cite{Horn-85}) and $\bR^*$ in \eqref{eq.thm.R*.1} reduces to the respective solutions in \cite{Zhang-10}-\cite{Loyka-17-2}. If all IPCs are inactive, then $\mu_{2k}=0$, $\bW_{\mu}= \sqrt{\mu_1}\bI$ and  $\bR^* =\bR_{WF}$, as it should be. Below, we will give explicit conditions when this is the case.

\vspace*{-.5\baselineskip}
\subsection{General properties}

Next, we explore some general properties of the capacity related to its unbounded growth with $P_T$ and its being strictly positive. It turns out that those properties induce a natural linear-algebraic structure for the set of channels of all users.

It is well-known that, without the IPC, $C(P_T)$ grows unbounded as $P_T$ increases, $C(P_T) \rightarrow \infty$ as $P_T \rightarrow \infty$ (assuming $\bW_1 \neq 0$). This, however, is not necessarily the case under the IPCs with all fixed $P_{Ik}$. The following proposition gives sufficient and necessary conditions when it is indeed the case.

\begin{prop}
\label{prop.C.inf}
Let $0 \le P_{Ik} < \infty$ be fixed for all $k$. Then, the capacity grows unbounded as $P_T$ increases, i.e. $C(P_T) \rightarrow \infty$ as $P_T \rightarrow \infty$, if and only if
\bal
\label{eq.prop.c.inf.1}
\bigcap_k \sN(\bW_{2k}) \notin \sN(\bW_1).
\eal
\vspace*{-1\baselineskip}
or, equivalently,
\bal
\sN\big(\sum_k \bW_{2k}\big) \notin \sN(\bW_1).
\eal
\end{prop}
\vspace*{-1\baselineskip}
\begin{proof}
See Appendix.
\end{proof}
\vspace*{-.3\baselineskip}

Let us make the following observations:

$\bullet$ Since the above conditions are both sufficient and necessary for the unbounded growth of the capacity, it gives the exhaustive characterization of all the cases where such growth is possible. In practical terms, those cases represent the scenarios where any high spectral efficiency is achievable given enough power budget.

$\bullet$ The unbounded growth of the capacity with $P_T$ depends only on $\sN(\sum_k \bW_{2k})$ and $\sN(\bW_1)$, all other details being irrelevant.

$\bullet$ It can be seen that the condition $\sN(\sum_k \bW_{2k}) \notin \sN(\bW_1)$ holds if $r(\sum_k \bW_{2k}) < r(\bW_1)$, and hence the capacity grows unbounded with $P_T$ under the latter condition.

$\bullet$ On the other hand, if $\sN(\sum_k \bW_{2k}) \in \sN(\bW_1)$, then very high spectral efficiency cannot be achieved even with unlimited power budget, due to the dominance of the IPCs. In particular, if $\bigcap_k \sN(\bW_{2k}) = \emptyset$ or, equivalently, $\sum_k \bW_{2k}>0$, then \eqref{eq.prop.c.inf.1} is impossible and the capacity stays bounded, even for infinite $P_T$ - the whole signaling space is dominated by IPCs in this case.

In the standard Gaussian MIMO channel without the IPC, $C=0$ if either $P_T=0$ or/and $\bW_1=0$, i.e. in a trivial way. On the other hand, in the same channel under the TPC and IPC, the capacity can be zero in non-trivial ways, as the following proposition shows. In practical terms, this characterizes the cases where interference constraints of primary users rule out any positive rate of a given user and, hence, spectrum sharing is not possible. To this end, let $\mathcal{K}_0=\{k: P_{Ik}=0\}$, i.e. a set of all primary users requiring no interference, $P_{Ik}=0$.

\begin{prop}
\label{prop.C0}
Consider the Gaussian MIMO channel under the TPC and IPC and let $P_T>0$, $\bW_1 \neq 0$. Its capacity is zero if and only if $P_{Ik}=0$ for some $k$ and
\bal
\label{eq.prop.C0.1}
\sN\big(\sum_{k\in \mathcal{K}_0} \bW_{2k}\big) \in \sN(\bW_1).
\eal
\end{prop}
\begin{proof}
See Appendix.
\end{proof}

Note that the condition $P_{Ik}=0$ is equivalent to zero-forcing transmission with respect to user $U_k$, i.e. the capacity is zero only if ZF transmission is required for at least one user; otherwise, $C>0$. The condition in \eqref{eq.prop.C0.1} cannot be satisfied if $r(\bW_1) > r(\sum_k \bW_{2k})$ and hence $C>0$ under the latter condition, which is also sufficient for unbounded growth of the capacity with $P_T$. This is summarized below.

\begin{cor}
\label{cor.C.inf}
If $r(\bW_1) > r(\sum_k \bW_{2k})$, then

1. $C \neq 0$ $\forall\ P_{Ik} \ge 0$ and $P_T>0$.

2. $C(P_T) \rightarrow\infty$ as  $P_T \rightarrow\infty$ $\forall\ P_{Ik} \ge 0$
\end{cor}

Thus, the condition $r(\bW_1) > r(\sum_k\bW_{2k})$ represents favorable propagation scenarios where spectrum sharing is possible for any $P_{Ik}$ and arbitrary large capacity can be attained given enough Tx power budget.

Unlike the standard WF where the TPC is always active, it can be inactive under the IPCs, which is ultimately due to the interplay of interference and power constraints. The following proposition explores this in some details. To this end, we call a constraint "redundant" if it can be omitted without affecting the capacity\footnote{"inactive" implies "redundant" but the converse is not true: for example, inactive TPC means $tr\bR^* < P_T$ and this implies $\mu_1=0$ (from complementary slackness) so that it is also redundant (can be omitted without affecting the capacity), but $\mu_1=0$ does not imply  $tr\bR^* < P_T$ since $tr\bR^* = P_T$ is also possible in some cases.}.

\begin{prop}
\label{prop.TPC.inact}
The TPC is redundant only if
\bal
\label{eq.prop.TPC.inact.1}
\mathcal{N}(\sum_k \bW_{2k}) \in \mathcal{N}(\bW_1)
\eal
and is active otherwise. In particular, it is active (for any $P_T$ and $P_{Ik}$) if $r(\bW_1)> r(\sum_k \bW_{2k})$, e.g. if $\bW_1$ is full-rank and $\sum_k\bW_{2k}$ is rank-deficient.
\end{prop}
\begin{proof}
See Appendix.
\end{proof}

\section{Full-rank solutions}

While Theorem \ref{thm.R*} establishes a closed-form solution for optimal covariance $\bR^*$ in the general case, it is expressed via dual variables $\mu_1, \mu_{2k}$ for which no closed-form solution is known in general so they have to be found numerically using \eqref{eq.thm.R*.2}. This limits insights significantly. In this section, we explore the cases when the optimal covariance $\bR^*$ is of full rank and obtain respective closed-form solutions. To this end, we set $K=1$, $\bW_2=\bW_{21}$, $P_I=P_{I1}$, $\mu_2=\mu_{21}$.  First, we consider an interference-limited regime, where the TPC is redundant and hence the IPC is active.

\begin{prop}
\label{prop.FR.IP}
Let $\bW_1, \bW_2 >0$ and $P_I$ be bounded as follows:
\bal
\label{eq.prop.FR.IP}
m\gl_1(\bW_2 &\bW^{-1}_1) - tr(\bW_2\bW^{-1}_1) < P_I\\ \notag
 &\le \frac{m}{tr(\bW_2^{-1})} (P_T+tr(\bW^{-1}_1)) - tr(\bW_2\bW^{-1}_1)
\eal
then $\mu_1=0$, i.e. the TPC is redundant, $\bR^*$ is of full-rank and is given by:
\bal
\label{eq.prop.FR.IP.2}
\bR^* = \mu_2^{-1}\bW_2^{-1} - \bW^{-1}_1
\eal
where $\mu_2^{-1}=m^{-1}(P_I+tr(\bW_2\bW^{-1}_1))$. The capacity can be expressed as
\bal
\label{eq.prop.FR.IP.3}
C= m\log((P_I&+tr(\bW_2\bW_1^{-1}))/m) +\log\frac{|\bW_1|}{|\bW_2|}
\eal
\end{prop}
\begin{proof}
See Appendix.
\end{proof}

Next, we explore the case where $\bW_2$ is of rank 1. This models the case when a primary user has a single-antenna receiver or when its channel is a keyhole channel, see e.g. \cite{Chizhik-02}\cite{Levin-08}.

\begin{prop}
\label{prop.W2.r1}
Let $\bW_1$ be of full rank and $\bW_2$ be of rank-1, so that $\bW_2=\gl_2\bu_2\bu_2^+$, where $\gl_2>0$ and $\bu_2$ are its active eigenvalue and eigenvector. If
\bal\notag
&P_I \ge P_{I,th} =  m^{-1}\gl_2(P_T + tr(\bW^{-1}_1)) -\gl_2\bu_2^+\bW^{-1}_1\bu_2\\
\label{eq.prop.W2-r1.1a}
&P_T> m\gl_1(\bW^{-1}_1) -tr(\bW^{-1}_1)
\eal
then the IPC is redundant, the optimal covariance is of full rank and is given by the standard WF solution,
\bal
\label{eq.prop.W2-r1.R*1}
\bR^* = \bR^*_{WF} = \mu_{WF}^{-1}\bI - \bW^{-1}_1
\eal
where $\mu_{WF}^{-1} = m^{-1}(P_T+tr(\bW^{-1}_1))$.

If
\bal
\label{eq.prop.W2-r1.3}
&\gl_2\gl_1(\bW^{-1}_1)-\gl_2\bu_2^+\bW^{-1}_1\bu_2 < P_I < P_{I,th},\\
\label{eq.prop.W2-r1.4}
 & P_T > m\gl_2^{-1}P_I + m\bu_2^+\bW^{-1}_1\bu_2 -tr(\bW^{-1}_1)
\eal
then the IPC and TPC are active, the optimal covariance is of full rank and is given by
\bal
\label{eq.prop.W2-r1.R*2}
\bR^* = \mu_1^{-1}\bI - \bW^{-1}_1 -\alpha\bu_2\bu_2^+
\eal
where $\alpha=\mu_1^{-1}-(\mu_1+\gl_2\mu_2)^{-1}$, and $\mu_1,\ \mu_2>0$ are
\bal
\label{eq.prop.W2-r1.5}
\begin{aligned}
\mu_1 &= (P_T-\gl_2^{-1} P_I -\bu_2^+\bW^{-1}_1\bu_2+tr(\bW^{-1}_1))^{-1} (m-1)\\
\mu_2 &= (P_I +\gl_2\bu_2^+\bW^{-1}_1\bu_2)^{-1}-\gl_2^{-1}\mu_1
\end{aligned}
\eal
\end{prop}
\vspace*{-1\baselineskip}
\begin{proof}
See Appendix.
\end{proof}

Note that the 1st two terms in \eqref{eq.prop.W2-r1.R*2} represent the standard WF solution while the last term is a correction due to the IPC, which is reminiscent of a partial null forming in an adaptive antenna array, see e.g. \cite{VanTrees-02}.

\section{Rank-1 Solutions}

In this section, we explore the case when $\bW_1$ is rank-one. As we show below, beamforming is optimal in this case. A practical appeal of this is due to its low-complexity implementation. Furthermore, rank-one $\bW_1$ is also motivated by single-antenna mobile units while the base station is equipped with multiple antennas, or when the MIMO propagation channel is of degenerate nature resulting in a keyhole effect, see e.g. \cite{Chizhik-02}\cite{Levin-08}.

We begin with the following result which bounds the rank of optimal covariance in any case.

\begin{prop}
\label{prop.rR}
If the TPC is active or/and $\bW_2$ is full-rank, then the rank of the optimal covariance $\bR^*$ of the problem (P1) in \eqref{eq.C.def} under the constraints in \eqref{eq.SR.2} is bounded as follows:
\bal
\label{eq.prop.rR}
r(\bR^*) \le r(\bW_1)
\eal
If the TPC is redundant and $\bW_2$ is rank-deficient, then there exists an optimal covariance $\bR^*$ of (P1) under the constraints in \eqref{eq.SR.2} that also satisfies this inequality.
\end{prop}
\begin{proof}
See Appendix.
\end{proof}

\begin{cor}
If $\bW_2$ is of full-rank or/and if the TPC is active, then the optimal covariance $\bR^*$ is of full-rank only if $\bW_1$ is of full-rank (i.e. rank-deficient $\bW_1$ ensures that $\bR^*$ is also rank-deficient).
\end{cor}

\begin{cor}
If $r(\bW_1)=1$, then $r(\bR^*)=1$, i.e. beamforming is optimal.
\end{cor}

Note that this rank (beamforming) property mimics the respective property for the standard WF. However, while signalling on the (only) active eigenvector of $\bW_1$ is optimal under the standard WF (no IPC), it is not so when the IPC is active, as the following result shows. To this end, let $\bW_1=\gl_1\bu_1\bu_1^+$, i.e. it is rank-1 with $\gl_1>0, \bu_1$ be the (only) active eigenvalue and eigenvector; $\gamma_I=P_I/P_T$ be the "interference-to-signal" ratio, and
\bal
\gamma_1 = \frac{\bu_1^+\bW_2^\dag\bu_1}{\bu_1^+(\bW_2^\dag)^2\bu_1},\ \gamma_2 = \bu_1^+\bW_2\bu_1
\eal
where $\bW_2^\dag$ is Moore-Penrose pseudo-inverse of $\bW_2$; $\bW_2^\dag=\bW_2^{-1}$ if $\bW_2$ is full-rank \cite{Horn-85}.

\begin{prop}
\label{prop.r1.IPC}
Let $\bW_1$ be rank-1.

1. If $\gamma_I < \gamma_1$, then the TPC is redundant and the optimal covariance can be expressed as follows
\bal
\label{eq.prop.r1.IPC.1}
\bR^* = P_I\frac{\bW_2^\dag\bu_1\bu_1^+\bW_2^\dag} {\bu_1^+\bW_2^\dag\bu_1}
\eal
\vspace*{-1\baselineskip}
The capacity is
\bal
\label{eq.prop.r1.IPC.2}
C = \log(1+\gl_1 \alpha P_T)
\eal
where $\alpha = \gamma_I \bu_1^+\bW_2^\dag\bu_1 < 1$.

2. If $\gamma_I\ge \gamma_2$, then the IPC is redundant and the standard WF solution applies: $\bR^* = P_T\bu_1\bu_1^+$. This condition is also necessary for the optimality of $P_T\bu_1\bu_1^+$ under the TPC and IPC when $\bW_1$ is rank-1. The capacity is as in \eqref{eq.prop.r1.IPC.2} with $\alpha=1$.

3. If $\gamma_1 \le \gamma_I < \gamma_2$, then both constraints are active. The optimal covariance is
\bal
\label{eq.prop.r1.IPC.3}
\bR^* = P_T\frac{\bW_{2\mu}^{-1}\bu_1\bu_1^+\bW_{2\mu}^{-1}} {\bu_1^+\bW_{2\mu}^{-2}\bu_1}
\eal
where $\bW_{2\mu}=\bI+\mu_2\bW_2$, and $\mu_2>0$ is found from the IPC: $tr(\bW_2\bR^*)=P_I$. The capacity is as in \eqref{eq.prop.r1.IPC.2} with
\bal
\alpha= (\bu_1^+\bW_{2\mu}^{-1}\bu_1)^2 |\bW_{2\mu}^{-1}\bu_1|^{-2} \le 1
\eal
with equality if and only if $\bu_1$ is an eigenvector of $\bW_2$.
\end{prop}
\begin{proof}
See Appendix.
\end{proof}

Note that the optimal signalling in case 1 is along the direction of $\bW_{2}^{\dag}\bu_1$ and not that of $\bu_1$ (unless $\bu_1$ is also an eigenvector of $\bW_{2}$), as would be the case for the standard WF with redundant IPC. In fact, $\bW_{2}^{\dag}$ plays a role of a "whitening" filter here. Similar observation applies to case 3, with $\bW_{2}$ replaced by $\bW_{2\mu}$. $\alpha$ in Proposition \ref{prop.r1.IPC} quantifies power loss due to enforcing the IPC; $\alpha=1$ means no power loss.

\section{Appendix}

\vspace*{-.5\baselineskip}
\subsection{Proof of Theorem \ref{thm.R*}}

Since the problem is convex and Slater's condition holds, the KKT conditions are both sufficient and necessary for optimality \cite{Boyd-04}. They take the following form:
\bal
\label{eq.T.KKT}
&-(\bI+\bW_1\bR)^{-1}\bW_1-\bM+\mu_1\bI+ \sum_k \mu_{2k}\bW_{2k} = 0\\
\label{eq.T.KKT.2}
&\bM\bR=0,\ \mu_1(tr(\bR)-P_T)=0,\notag \\
&\mu_{2k}(tr(\bW_{2k}\bR)-P_{Ik})=0,\\
&\bM\ge 0,\ \mu_1\ge 0,\ \mu_{2k}\ge 0\\
&tr(\bR) \le P_T,\ tr(\bW_{2k}\bR) \le P_{Ik},\ \bR\ge 0
\eal
where $\bM$ is Lagrange multiplier responsible for the positive semi-definite constraint $\bR\ge 0$. We consider first the case of full-rank $\bW_{\mu}$ (i.e. either $\mu_1>0$ or/and $\sum_k\mu_{2k}\bW_2>0$), so that $\bW_{\mu}^{\dag}=\bW_{\mu}^{-1}$. Let us introduce new variables: $\tilde{\bR}= \bW_{\mu}\bR\bW_{\mu}$, $\tilde{\bW_1}= \bW_{\mu}^{-1}\bW_1\bW_{\mu}^{-1}$, $\tilde{\bM}= \bW_{\mu}^{-1}\bM\bW_{\mu}^{-1}$. It follows that $\tilde{\bM}\tilde{\bR}=0$ and \eqref{eq.T.KKT} can be transformed to
\bal
&(\bI+\tilde{\bW_1}\tilde{\bR})^{-1}\tilde{\bW_1} +\tilde{\bM} = \bI
\eal
for which the solution is
\bal
\label{eq.T.KKT.Rtil}
\tilde{\bR} = (\bI-\tilde{\bM})^{-1} - \tilde{\bW}_1^{-1} = (\bI- \tilde{\bW}_1^{-1})_+
\eal
(this can be established in the same way as for the standard WF). Transforming back to the original variables results in \eqref{eq.thm.R*.1}. \eqref{eq.thm.R*.2} are complementary slackness conditions in \eqref{eq.T.KKT.2}; \eqref{eq.C} follows, after some manipulations, by using $\bR^*$ of \eqref{eq.thm.R*.1} in $C(\bR)$.

The case of singular $\bW_{\mu}$ is more involved. It implies $\mu_1=0$ so that $\bW_{\mu}=(\sum_k \mu_{2k}\bW_{2k})^{\frac{1}{2}}$. It follows from the KKT condition in \eqref{eq.T.KKT} that, for the redundant TPC ($\mu_1=0$),
\bal
\label{eq.thm.Inac.TPC.1}
\bQ_1(\bI+\bQ_1\bR\bQ_1)^{-1}\bQ_1+\bM=\sum_k \mu_{2k}\bW_{2k}
\eal
where $\bQ_1=\bW_1^{1/2}$. Let $\bx\in \mathcal{N}(\sum_k \mu_{2k}\bW_{2k})$, i.e. $\sum_k \mu_{2k}\bW_{2k}\bx=0$, then
\bal
\bx^+\bQ_1(\bI+\bQ_1\bR\bQ_1)^{-1}\bQ_1\bx +\bx^+\bM\bx=0
\eal
so that $\bx^+\bM\bx=0$ and $\bQ_1\bx=0$, since $\bM\ge0$ and $\bI+\bQ_1\bR\bQ_1 > 0$. Thus, $\mathcal{N}(\sum_k \mu_{2k}\bW_{2k}) \in \mathcal{N}(\bQ_1)=\mathcal{N}(\bW_1)$ and $\mathcal{N}(\sum_k \mu_{2k}\bW_{2k}) \in \mathcal{N}(\bM)$, i.e.
\bal
\label{eq.thm.Inac.TPC.3}
\mathcal{N}(\sum_k \mu_{2k}\bW_{2k}) \in \mathcal{N}(\bW_1) \cap \mathcal{N}(\bM)
\eal
and this condition is also necessary for the TPC to be redundant. Further notice that
\bal
\label{eq.thm.Inac.TPC.3a}
\mathcal{N}(\sum_k \mu_{2k}\bW_{2k})= \bigcap_{k\in \mathcal{K}_+} \mathcal{N}(\bW_{2k}) = \mathcal{N}(\sum_{k\in \mathcal{K}_+} \bW_{2k})
\eal
where $\mathcal{K}_+=\{k: \mu_{2k}>0\}$ is the set of users with active IPCs.
Let $\bW_2 = \sum_k \mu_{2k}\bW_{2k}$. Using \eqref{eq.thm.Inac.TPC.1}, \eqref{eq.thm.Inac.TPC.3} and introducing new variables
\bal\notag
\bLam_2&=\bU_2^+\bW_2\bU_2,\ \tilde{\bR}=\bU_2^+\bR\bU_2,\\ \tilde{\bQ}_1&=\bU_2^+\bQ_1\bU_2,\ \tilde{\bM}=\bU_2^+\bM\bU_2,
\eal
where $\bU_2$ is a unitary matrix of eigenvectors of $\bW_{2}$, one obtains
\bal\notag
&\bLam_2=\left(
          \begin{array}{cc}
            \bLam_{2+} & 0 \\
            0 & 0 \\
          \end{array}
        \right),
\tilde{\bQ}_1=\left(
          \begin{array}{cc}
            \bQ_{1+} & 0 \\
            0 & 0 \\
          \end{array}
        \right),\\
&\tilde{\bM}=\left(
          \begin{array}{cc}
            \bM_{+} & 0 \\
            0 & 0 \\
          \end{array}
        \right),
\tilde{\bR}=\left(
          \begin{array}{cc}
            \bR_{+} & \bR_{12}\\
            \bR_{21} & \bR_{22} \\
          \end{array}
        \right)
\eal
where $\bLam_{2+}>0$ is a diagonal matrix of strictly positive eigenvalues of $\bW_2$, so that \eqref{eq.thm.Inac.TPC.1} can be transformed to
\bal
\label{eq.Inac.TPC.2}
\bQ_{1+}(\bI+\bQ_{1+}\bR_+\bQ_{1+})^{-1} \bQ_{1+} + \bM_+= \bLam_{2+}>0
\eal
Using
\bal
\bQ_{1+}(\bI+\bQ_{+1}\bR_+\bQ_{1+})^{-1} \bQ_{1+} = (\bI+\bW_{1+}\bR_+)^{-1} \bW_{1+}
\eal
where $\bW_{1+}= \bQ_{1+}^2$ and adopting \eqref{eq.T.KKT.Rtil}, \eqref{eq.thm.R*.1}, one obtains
\bal
\label{eq.T.KKT.Rtil+}
\bR_+ = \bLam_{2+}^{-\frac{1}{2}}(\bI- \bLam_{2+}^{\frac{1}{2}}\bW_{1+}^{-1}\bLam_{2+}^{\frac{1}{2}})_+ \bLam_{2+}^{-\frac{1}{2}}
\eal
Since only $\bR_+$ affects the capacity,
one can set, without loss of optimality, $\bR_{22}=0$, $\bR_{12}=0$, $\bR_{21}=0$, and transform \eqref{eq.T.KKT.Rtil+} to
\bal
\label{eq.T.KKT.Rtil2}
\tilde{\bR} = (\bLam_{2}^{\dag})^{\frac{1}{2}}(\bI- \bLam_{2}^{\frac{1}{2}}\tilde{\bW}_{1}^{-1}\bLam_{2}^{\frac{1}{2}})_+ (\bLam_{2}^{\dag})^{\frac{1}{2}}
\eal
and hence, as desired,
\bal
\label{eq.T.KKT.R2}
\bR= \bU_2\tilde{\bR}\bU_2^+ = \bW_{\mu}^{\dag}(\bI- \bW_{\mu}\bW_{1}^{-1}\bW_{\mu})_+ \bW_{\mu}^{\dag}
\eal

\subsection{Proof of Proposition \ref{prop.C.inf}}

To prove the "if" part, observe that $\bigcap_k\sN(\bW_{2k}) \notin \sN(\bW_1)$ implies $\exists \bu: \bW_{2k}\bu=0\ \forall k, \bW_1\bu \neq 0$.
Now set $\bR=P_T\bu\bu^+$, for which $tr(\bR) =P_T, tr(\bW_{2k}\bR)=0\ \forall k$, so it is feasible for any $P_T, P_{Ik}$. Furthermore,
\bal
C \ge C(\bR) = \log(1+P_T\bu^+\bW_1\bu) \rightarrow \infty
\eal
as $P_T \rightarrow \infty$, since $\bu^+\bW_1\bu >0$.

Next, we will need the following technical result, which will also establish the last claim.
\begin{lemma}
The following holds:
\bal
\bigcap_k \sN(\bW_{2k}) = \sN\big(\sum_k \bW_{2k}\big)
\eal
\begin{proof}
Since $\bW_{2k}\ge 0$, $\bW_{2k}\bu =0$ is equivalent to $\bu^+\bW_{2k}\bu =0$ \cite{Zhang-99}, so that the following equivalence holds:
\bal\notag
\bu \in \sN\big(\sum_k \bW_{2k}\big) &\leftrightarrow \sum_k \bW_{2k}\bu =0 \leftrightarrow \sum_k \bu^+\bW_{2k}\bu =0 \\ \notag
&\leftrightarrow \bu^+\bW_{2k}\bu =0\ \forall k \leftrightarrow \bW_{2k}\bu =0\ \forall k \\ \notag
&\leftrightarrow \bu \in \bigcap_k \sN(\bW_{2k})
\eal
\end{proof}
\end{lemma}

To prove the "only if" part, let $\bW_2=\sum_k\bW_{2k},\ P_I=\sum_k P_{Ik}$ and assume that $\sN(\bW_2) \in \sN(\bW_1)$. This implies that $\sR(\bW_1) \in \sR(\bW_2)$ (since $\sR(\bW)$ is the complement of $\sN(\bW)$ for Hermitian $\bW$). Let
\bal
\bW_k = \bU_{k+} \bLam_{k}\bU_{k+}^+,\ k=1,2
\eal
where $\bU_{k+}$ is a semi-unitary matrix of active eigenvectors of $\bW_k$ and diagonal matrix $\bLam_{k}$ collects its strictly-positive eigenvalues. Notice that, from the IPC,
\bal\notag
P_I &\ge tr(\bW_2\bR) = tr(\bLam_2\bU_{2+}^+\bR\bU_{2+})\\
 &\ge \gl_{r_2} tr(\bU_{2+}^+\bR\bU_{2+})
\eal
where $\gl_{r_2}>0$ is the smallest positive eigenvalue of $\bW_2$, so that
\bal
 \gl_1(\bU_{2+}^+\bR\bU_{2+}) \le P_I/\gl_{r_2} < \infty
\eal
for any $P_T$. On the other hand, $\sR(\bW_{1}) \in \sR(\bW_2)$ implies $span\{\bU_{1+}\} \in span\{\bU_{2+}\}$ and hence
\bal
 \gl_1(\bU_{1+}^+\bR\bU_{1+}) \le \gl_1(\bU_{2+}^+\bR\bU_{2+}) \le P_I/\gl_{r_2} < \infty
\eal
so that
\bal\notag
C(P_T) &= \log|\bI+ \bLam_1\bU_{1+}^+\bR^*\bU_{1+}|\\ \notag
&= \sum_i \log(1+\gl_i(\bLam_1\bU_{1+}^+\bR^*\bU_{1+}))\\ \notag
&\le m \log(1+\gl_1(\bW_1)\gl_1(\bU_{1+}^+\bR^*\bU_{1+}))\\
 &\le m \log(1+\gl_1(\bW_1)P_I/ \gl_{r_2}) < \infty
\eal
is bounded for any $P_T$, as required.

\subsection{Proof of Proposition \ref{prop.C0}}

To prove the "if" part, observe that $tr(\bW_{2k}\bR)=P_{Ik}=0$ implies that $\bW_{2k}\bR=0$ (since $\gl_i(\bW_{2k}\bR) \ge 0$) so that $\sR(\bR) \in \sN(\bW_{2k})$ for any $k\in  \mathcal{K}_0$ and hence $\sR(\bR) \in \sN(\sum_{k\in \mathcal{K}_0} \bW_{2k})$. Under the above condition, this implies $\sR(\bR) \in \sN(\bW_1)$ and hence $\bW_1\bR=0$ so that $\log|\bI+\bW_1\bR|=0$ for any feasible $\bR$. Hence, $C=0$.

To prove the "only if" part, assume first that $P_{Ik}>0$ for all $k$ and set $\bR= \min\{P_T, p\}\bI/m$, where $p= \min_k \{P_{Ik}/\gl_1(\bW_{2k})\}$. Note that $\bR$ is feasible: $tr(\bR) \le P_T$ and $tr(\bW_{2k}\bR) \le P_{Ik}$. Furthermore,
\bal
C \ge \log|\bI+\bW_1\bR| >0
\eal
and hence $P_{Ik}=0$ for some $k$ is necessary for $C=0$. To show that \eqref{eq.prop.C0.1} is necessary as well, assume that it does not hold, which implies that $\exists \bu: \bW_{2k}\bu=0\ \forall k \in \mathcal{K}_0,\ \bW_1\bu \neq 0$. Now set $\bR=p\bu\bu^+$, where $p=P_T$ if $P_{Ik}=0\ \forall k$; otherwise, $p= \min\{ P_T, p_1\}$, where $p_1=\min_{k \notin \mathcal{K}_0}\{P_{Ik}/\gl_1(\bW_{2k})\}$. Notice that, for this $\bR$, $tr(\bR) \le P_T, tr(\bW_{2k}\bR) \le P_{Ik}\ \forall k$, so it is feasible and
\bal
C \ge \log|\bI+\bW_1\bR| = \log(1+p\bu^+\bW_1\bu) >0
\eal
so that \eqref{eq.prop.C0.1} is necessary for $C=0$.

\subsection{Proof of Proposition \ref{prop.TPC.inact}}

Use \eqref{eq.thm.Inac.TPC.3} and \eqref{eq.thm.Inac.TPC.3a}, and note that these conditions are necessary for the TPC to be redundant (since the KKT conditions are necessary for optimality and $\mu_1=0$ is also necessary for the TPC to be redundant). Now, if $r(\bW_1)> r(\sum_k \bW_{2k})$, then
\bal\notag
\dim(\mathcal{N}(\sum_k \bW_{2k})) &=m-r(\sum_k \bW_{2k}) > m-r(\bW_1)\\
& =\dim(\mathcal{N}(\bW_1))
\eal
where $\dim(\mathcal{N})$ is the dimensionality of $\mathcal{N}$, and hence \eqref{eq.prop.TPC.inact.1} is impossible so that the TPC is active.

\subsection{Proof of Proposition \ref{prop.FR.IP}}

When the TPC is redundant, $\mu_1=0$ and \eqref{eq.prop.FR.IP.2} with
\bal
\mu_2^{-1}=m^{-1}(P_I+tr(\bW_2\bW^{-1}_1))
\eal
follow from \eqref{eq.thm.R*.1}. 1st condition in \eqref{eq.prop.FR.IP} follows from $\bR^*>0$, which is equivalent to
\bal
\gl_m(\bW_2^{-1}\bW_1) > \mu_2= m(P_I+tr(\bW_2\bW^{-1}_1))^{-1}
\eal
2nd condition in \eqref{eq.prop.FR.IP} ensures that the TPC is redundant for sufficiently large Tx power: $\mu_1=0$ and $tr(\bR^*) \le P_T$.

\subsection{Proof of Proposition \ref{prop.W2.r1}}

Start with the matrix inversion Lemma to obtain
\bal\notag
(\mu_1\bI +\mu_2\gl_2\bu_2\bu_2^+)^{-1} = ((\mu_1 &+\mu_2\gl_2)^{-1}- \mu_1^{-1})\bu_2\bu_2^+\\
&+\mu_1^{-1}\bI
\eal
so that \eqref{eq.prop.W2-r1.R*2} follows from \eqref{eq.thm.R*.1}. Since $\bW_1$ is full-rank and $\bW_2$ is rank-1, it follows that the TPC is always active, $\mu_1>0$ and $tr(\bR)=P_T$, from which one obtains
\bal
\label{eq.prop.W2-r1.R*.4}
m\mu_1^{-1} -\alpha - tr(\bW^{-1}_1) = P_T
\eal
When the IPC is active, $tr(\bW_2\bR)=P_I$, it follows that
\bal
\label{eq.prop.W2-r1.R*.5}
\gl_2(\mu_1+\mu_2\gl_2)^{-1}= P_I +\gl_2\bu_2^+\bW^{-1}_1\bu_2
\eal
Solving \eqref{eq.prop.W2-r1.R*.4} and \eqref{eq.prop.W2-r1.R*.5} for $\mu_1$, one obtains 1st equality in \eqref{eq.prop.W2-r1.5}; using it in \eqref{eq.prop.W2-r1.R*.5} results in 2nd equality in \eqref{eq.prop.W2-r1.5}. \eqref{eq.prop.W2-r1.4} and 1st inequality in \eqref{eq.prop.W2-r1.3} ensure that $\bR^* >0$, since
\bal
\label{eq.prop.W2-r1.R*.6}
\mu_1^{-1} > \gl_1(\bW^{-1}_1) +\alpha \ge \gl_1(\bW^{-1}_1+\alpha\bu_2\bu_2^+)
\eal
where 1st inequality is due to 1st inequality in \eqref{eq.prop.W2-r1.3} and \eqref{eq.prop.W2-r1.R*.5} while 2nd inequality is from $\gl_1(\bA+\bB) \le \gl_1(\bA)+\gl_1(\bB)$ where $\bA,\ \bB$ are Hermitian matrices (see e.g. \cite{Horn-85}). It follows from \eqref{eq.prop.W2-r1.R*.6} that $\mu_1^{-1}\bI > \bW^{-1}_1 +\alpha\bu_2\bu_2^+$ and hence $\bR^*>0$, and that $\mu_1>0$, as required. 2nd inequality in \eqref{eq.prop.W2-r1.3} ensures that the IPC is active, $\mu_2>0$.

To obtain \eqref{eq.prop.W2-r1.R*1}, observe that $\bR_{WF}$ is feasible under \eqref{eq.prop.W2-r1.1a}:
\bal
tr(\bR_{WF}) =P_T,\ tr(\bW_2\bR_{WF}) \le P_I,\ \bR_{WF} >0.
\eal
Since it is a solution without the IPC (as the standard full-rank WF solution), it is also optimal under the IPC.

\subsection{Proof of Proposition \ref{prop.rR}}

We consider first the case when $\bW_{\mu}$ is full-rank, i.e. when either the TPC is active, $\mu_1>0$, or/and $\bW_2>0$. It follows from \eqref{eq.T.KKT} that
\bal
(\bI+\bW_1\bR^*)^{-1}\bW_1\bR^* = \bW_{\mu}^2\bR^*
\eal
so that, since $(\bI+\bW_1\bR)$ and $\bW_{\mu}^2$ are full-rank,
\bal\notag
r(\bR^*) &= r(\bW_{\mu}^2\bR^*) = r(\bW_1\bR^*)\\
 &\le \min\{r(\bW_1),r(\bR^*)\} \le r(\bW_1)
\eal

The case of rank-deficient $\bW_{\mu}^2$ (i.e. when $\mu_1=0$ and $\bW_2$ is rank-deficient) is more involved. In this case, it follows from Proposition \ref{prop.TPC.inact} that $\mathcal{N}(\bW_2) \in \mathcal{N}(\bW_1)$ and hence $\mathcal{R}(\bW_1)\in \mathcal{R}(\bW_2)$ (if $\bW$ is Hermitian, $\mathcal{R}(\bW)$ is the complement of $\mathcal{N}(\bW)$), from which the following equivalency can be established, which is instrumental in the proof.

\begin{prop}
\label{prop.NW2}
If $\bW_2$ is rank-deficient and the TPC is redundant for the problem (P1) in \eqref{eq.C.def} under the constraint in \eqref{eq.SR.2}, then (P1) has the same value as the following problem (P2):
\bal
\label{eq.C.tilda}
(P2):\ \max_{\tilde{\bR} \ge 0} \tilde{C}(\tilde{\bR})\ \mbox{s.t.}\ tr(\tilde{\bLam}_2\tilde{\bR}) \le P_I,\ tr(\tilde{\bR}) \le P_T
\eal
where $\tilde{C}(\tilde{\bR})= |\bI+ \tilde{\bW}_1\tilde{\bR}|$, $\tilde{\bW}_1 = \bU_{2+}^+\bW_1\bU_{2+}$, $\tilde{\bLam}_2 = \bU_{2+}^+\bW_2\bU_{2+}>0$ is a diagonal matrix of strictly-positive eigenvalues of $\bW_2$ and $\bU_{2+}$ is a semi-unitary matrix whose columns are the corresponding active eigenvectors of $\bW_2$. Furthermore, an optimal covariance $\bR^*$ of (P1) can be expressed as follows:
\bal
\label{eq.R*2}
\bR^* = \bU_{2+}\tilde{\bR^*}\bU_{2+}^+
\eal
where $\tilde{\bR^*}$ is a solution of \eqref{eq.C.tilda}:
\bal
\label{eq.R.tild}
\tilde{\bR}^* = \tilde{\bLam}_2^{-\frac{1}{2}} (\mu_2^{-1}\bI-\tilde{\bLam}_2^{\frac{1}{2}}\tilde{\bW}_1^{-1} \tilde{\bLam}_2^{\frac{1}{2}})_+ \tilde{\bLam}_2^{-\frac{1}{2}}
\eal
and $\mu_2> 0$ is found from the IPC:
\bal
tr (\mu_2^{-1}\bI- \tilde{\bLam}_2^{\frac{1}{2}}\tilde{\bW}_1^{-1} \tilde{\bLam}_2^{\frac{1}{2}})_+ = P_I
\eal
\end{prop}
\begin{proof}
Let $\bR^*$ and $\tilde{\bR}^*$  be the solutions of (P1) and (P2) under the stated conditions and let $\bP_2=\bU_{2+}\bU_{2+}^+$ be a projection matrix on the space spanned by the active eigenvectors of $\bW_2$, i.e. on $\mathcal{R}(\bW_2)$. Note that $\bP_2\bW_k\bP_2=\bW_k$, $k=1,2$, since $\mathcal{R}(\bW_1)\in \mathcal{R}(\bW_2)$ under the stated conditions. Define $\tilde{\bR}'= \bU_{2+}^+\bR^*\bU_{2+}$ and observe that
\bal\notag
P_T &\ge tr(\bR^*) \ge tr(\tilde{\bR}'),\\
P_I &\ge tr(\bW_2\bR^*) = tr(\bP_2\bW_2\bP_2\bR^*) = tr(\tilde{\bLam}_2\tilde{\bR}')
\eal
so that $\tilde{\bR}'$ is feasible for (P2) and hence
\bal\notag
\tilde{C}(\tilde{\bR}^*)&\ge \tilde{C}(\tilde{\bR}')= \log|\bI+\tilde{\bW}_1\tilde{\bR}'|\\
&= \log|\bI+\bP_2\bW_1\bP_2\bR^*|
 =C(\bR^*)
\eal
On the other hand, let $\bR'= \bU_{2+}\tilde{\bR}^*\bU_{2+}^+$ and observe that
\bal\notag
P_T &\ge tr(\tilde{\bR}^*) = tr(\bR'),\\
P_I &\ge tr(\tilde{\bLam}_2\tilde{\bR}^*) = tr(\bP_2\bW_2\bP_2\bR') = tr(\bW_2\bR')
\eal
so that $\bR'$ is feasible for (P1) and hence
\bal
C(\bR^*)\ge C(\bR') =\log|\bI+\tilde{\bW}_1\tilde{\bR}^*|=\tilde{C}(\tilde{\bR}^*)
\eal
and finally $C(\bR^*)= \tilde{C}(\tilde{\bR}^*)$, $\mu_1=0$ (since the TPC is redundant for the original problem and hence for both problems) and the desired result follows.
\end{proof}

Note that Proposition \ref{prop.NW2} establishes the optimality of projecting all matrices on the sub-space $\mathcal{R}(\bW_2)$ and solving the projected problem instead, if the TPC is not active and $\bW_2$ is rank-deficient, i.e. if $\bW_{\mu}$ is rank-deficient.

Adopting the KKT condition in \eqref{eq.T.KKT} to the problem in \eqref{eq.C.tilda}, one obtains:
\bal
(\bI+\tilde{\bW}_1\tilde{\bR}^*)^{-1}\tilde{\bW}_1\tilde{\bR}^* = \mu_2\tilde{\bLam}_2\tilde{\bR}^*
\eal
so that
\bal\notag
r(\tilde{\bR}^*) &= r(\tilde{\bLam}_2\tilde{\bR}^*) = r(\tilde{\bW}_1\tilde{\bR}^*) \le \min(r(\tilde{\bW}_1),r(\tilde{\bR}^*))\\
    &\le r(\tilde{\bW}_1) \le r(\bW_1)
\eal
and, from \eqref{eq.R*2}, $r(\bR^*)=r(\tilde{\bR}^*)$, so that $r(\bR^*) \le r(\bW_1)$, as desired.

\subsection{Proof of Proposition \ref{prop.r1.IPC}}

To establish these results, we need the following technical Lemma, which can be established via the standard continuity argument.
\begin{lemma}
\label{lemma.I-W}
Let $\bW=\gl\bu\bu^+$ be rank-one positive semi-definite matrix, $\gl>0$. Then,
\bal
(\bI-\bW^{-1})_+=(1-\gl^{-1})_+\bu\bu^+
\eal
\end{lemma}
Note that the $(\cdot)_+$ operator eliminates all singular modes of $\bW$ and hence its singularity is not a problem, which is somewhat similar to using pseudo-inverse for a singular matrix.

To prove the 1st case, we assume that $\bW_2$ is not singular and discuss the singular case later. Setting $\bW=\bW_{2\mu}^{-\frac{1}{2}}\bW_1\bW_{2\mu}^{-\frac{1}{2}}$ and applying this Lemma to $(\bI-(\bW_{2\mu}^{-\frac{1}{2}}\bW_1\bW_{2\mu}^{-\frac{1}{2}})^{-1})_+$
in \eqref{eq.thm.R*.1}, one obtains $\bR^*$ as in \eqref{eq.prop.r1.IPC.1}, after some manipulations, with $\bW_2^\dag=\bW_2^{-1}$. The condition $\gamma_I< \gamma_1$ ensures that the TPC is redundant, so that $\mu_1=0$ and hence $\bW_{2\mu}=\mu_2\bW_2>0$, $tr(\bW_2\bR^*)=P_I$ (since the IPC is active).

If $\bW_2$ is singular and the TPC is redundant, then one can project all matrices on $\mathcal{R}(\bW_2)$ and solve the projected problem instead without loss of optimality, as was shown in Proposition \ref{prop.NW2}. After some manipulations, this can be shown to result in using the pseudo-inverse instead of the inverse of $\bW_2$.

To prove the $\gamma_I\ge \gamma_2$ case, note that, under this condition, $\bR^* = P_T\bu_1\bu_1^+$ is feasible under the joint constraint (TPC+IPC). Since it is also optimal without the IPC, it has to be optimal under the joint constraints as well. This proves the "if" part. To prove the "only if" (necessary) part, observe that if $P_T \bu_1^+\bW_2\bu_1 > P_I$, then $\bR^* = P_T\bu_1\bu_1^+$ is not feasible and hence cannot be optimal under the IPC.

To prove the last case, $\gamma_1 \le \gamma_I < \gamma_2$, use \eqref{eq.thm.R*.1} and note that both constraints are now active (since neither \eqref{eq.prop.r1.IPC.1} nor $\bR^* = P_T\bu_1\bu_1^+$ are feasible under the stated conditions). Applying Lemma \ref{lemma.I-W} as in the 1st case, one obtains \eqref{eq.prop.r1.IPC.3} after some manipulations.

\end{document}